\documentclass[11pt]{article}

\usepackage{times}
\usepackage{graphicx}
\usepackage{mathrsfs,amssymb,amsmath,textcomp,amsfonts,amsthm,fullpage}

\usepackage{multirow,amsmath, dsfont, amscd, latexsym,color,amssymb,setspace}
\usepackage{ifthen}
\usepackage[noend]{algorithmic}
\usepackage{enumerate}
\usepackage{bm}

\newtheorem{Thm}{Theorem}
\newtheorem{Lem}[Thm]{Lemma}
\newtheorem{Cor}[Thm]{Corollary}

\newtheorem{Claim}[Thm]{Claim}

\newtheorem{Def}[Thm]{Definition}

\def\bull{\vrule height .9ex width .8ex depth -.1ex }

\newenvironment{Proof}{\medbreak
\noindent {\bf Proof:~}}{\unskip\nobreak\hfill\hskip 2em \bull\par\medbreak}

\newcommand{\eat}[1]{}

\newcommand{\rgta}{\ensuremath{\rightarrow}}

\newcommand{\samp}{\ensuremath{\leftarrow}}

\DeclareMathOperator{\wt}{wt}
\DeclareMathOperator{\rk}{rank}

\newcommand{\zo}{\{0,1\}}

\newcommand{\pmo}{\ensuremath \{ \pm 1\}}

 \newcommand{\commented}{no}

\ifthenelse{\equal{\commented}{yes}}{
\newcommand{\pnote}[1]{\footnote{{\bf [[Parik: {#1}\bf ]] }}}
\newcommand{\snote}[1]{\footnote{{\bf [[Salil: {#1}\bf ]] }}}
\newcommand{\ynote}[1]{\footnote{{\bf [[Yuan: {#1}\bf ]] }}}
}

\ifthenelse{\equal{\commented}{no}}{
\newcommand{\pnote}[1]{}
\newcommand{\snote}[1]{}
\newcommand{\ynote}[1]{}
}

\newcommand{\ignore}[1]{}

\newcommand{\E}{{\bf E}}

\newcommand{\Rej}{\mathrm{Rej}}

\newcommand{\R}{\mathbb R}
\newcommand{\N}{\mathbb N}
\newcommand{\F}{\mathbb F}

\newcommand{\eps}{\varepsilon}

\newcommand{\polylog}{\mathrm{polylog}}

\newcommand{\Cay}{\mathrm{Cay}}

\newcommand{\bx}{{\bar{x}}}
\newcommand{\by}{{\bar{y}}}
\newcommand{\be}{{\bar{e}}}

\newcommand{\calC}{{\cal C}}

\newcommand{\calE}{{\cal E}}

\newcommand{\calV}{{\cal V}}
\newcommand{\calS}{{\cal S}}
\newcommand{\calT}{{\cal T}}

\newcommand{\calD}{{\cal D}}
\newcommand{\calA}{{\cal A}}
\newcommand{\calB}{{\cal B}}
\newcommand{\calG}{{\cal G}}

\newcommand{\cosC}{{\calV/\calC}}
\newcommand{\bcalE}{{\bar{\calE}}}

\newcommand{\LTC}{Locally Testable Code}
\newcommand{\CG}{Cayley Graph}
\newcommand{\etal}{et al.}

\newcounter{this-list}

\begin{document}

\begin{titlepage}

\title{Locally Testable Codes and Cayley Graphs}
\author{Parikshit Gopalan\thanks{Microsoft Research Silicon Valley. Email: \texttt{parik@microsoft.com}.} \and Salil Vadhan\thanks{School of Engineering and Applied Sciences, Harvard University. E-mail: \texttt{salil@seas.harvard.edu}.  Work done in part when on leave as a Visiting Researcher at Microsoft Research Silicon Valley and a Visiting Scholar at Stanford University.  Supported in part by NSF grant CCF-1116616 and US-Israel BSF grant 2010196.} \and Yuan Zhou\thanks{Computer Science Department, Carnegie Mellon University. E-mail: \texttt{yuanzhou@cs.cmu.edu}. Supported in part by a grant from the Simons Foundation (Award Number 252545). Work done in part when visiting Microsoft Research Silicon Valley.}}

\maketitle

\begin{abstract}
We give two new characterizations of ($\F_2$-linear) locally testable error-correcting
codes in terms of Cayley graphs over $\F_2^h$:
\begin{enumerate}
\item A locally testable code is equivalent to a Cayley graph over $\F_2^h$ whose set of generators is significantly larger than
$h$ and has no short linear dependencies, but yields a shortest-path metric that embeds into $\ell_1$ with constant distortion.
This extends and gives a converse to a result of Khot and Naor
(2006), which showed that codes with large dual distance imply Cayley
graphs that have no low-distortion embeddings into $\ell_1$. 

\item A locally testable code is equivalent to a Cayley graph over
  $\F_2^h$ that has significantly more than $h$ eigenvalues near 1,
  which have no short linear dependencies among them and which ``explain'' all of the large eigenvalues.
This extends and  gives a converse to a recent construction of
Barak et al. (2012), which showed that locally testable codes imply Cayley graphs
that are small-set expanders but have many large eigenvalues.
\end{enumerate}
\end{abstract}
\vfill
\thispagestyle{empty}
\end{titlepage}

\section{Introduction}

\newcommand{\Close}{\calD_{\mathrm{close}}}
\newcommand{\Far}{\calD_{\mathrm{far}}}

In this work, we show that {\em locally testable codes} are
equivalent to {\em Cayley graphs} with certain properties, thereby providing 
a new perspective from which to approach long-standing open problems about the
achievable parameters of locally testable codes.

Before describing these results, we review the basics of both locally testable codes
and Cayley graphs.

\subsection{Locally Testable Codes}

Informally, a {\em locally testable code (LTC)} is an error-correcting code in which one can distinguish received words that are in the code from those that are far from the code by a randomized test that probes only a few coordinates of the received word.  Local testing algorithms for algebraic error-correcting codes (like the Hadamard code and the Reed--Muller code) were developed in the literature on program testing~\cite{BlumLuRu93,RubinfeldSu96}, inspired
the development of the field of property testing~\cite{GoldreichGoRo98}, and played
a key role in the constructions of multi-prover interactive proofs and the proof of the PCP Theorem~\cite{BabaiFoLu91,FeigeGoLoSaSz96,BabaiFoLeSz91,AroraSa98,AroraLuMoSuSz98}.  Indeed, they are considered to be the ``combinatorial core'' of PCPs (cf., \cite{GoldreichSu06,BenSassonGoHaSuVa06}), and thus understanding what is possible and impossible with locally testable codes can point the way to a similarly improved understanding of PCPs.  See the surveys \cite{Trevisan04,Goldreich11-ltc,BenSasson10}.  

We focus on the commonly studied case of linear codes over $\F_2$.  Thus a {\em code} is specified by a linear subspace $\calC\subset \F_2^n$. 
$n$ is called the {\em blocklength} of the code, and $k=\dim(\calC)$ is its {\em rate}.  The {\em minimum distance} is 
$d=\min_{x\neq y\in \calC} d(x,y) = \min_{x\in \calC-\{0,\}} |x|$, where $d(\cdot,\cdot)$ denotes Hamming distance and $|\cdot|$ denotes Hamming weight.  

A {\em local tester} for $\calC$ is a randomized algorithm $T$ that,
when given oracle access to a {\em received word} $r\in \F_2^n$, makes
at most a small number $q$ of queries to symbols of $r$ and
accepts or rejects.  If $r\in \calC$, then $T^r$ should accept with high probability (completeness), and if $r$ is ``far'' from $\calC$ in Hamming distance, then $T^r$ should reject with high probability (soundness).  It was shown in \cite{BenSassonHaRa05} 
that any tester for a
linear code can be converted into one with the following structure: the tester $T$ randomly samples a string $\alpha\gets \calD$ and accepts if $\alpha\cdot r = 0$, where $\calD=\calD_T$ is some distribution on the {\em dual code} $\calC^\perp = \{ \alpha\in \F_2^n : \alpha\cdot c=0 \forall\ c\in \calC \}$.  In particular, such a tester has perfect completeness (accepts with probability 1 if $r\in \calC$).  We say the tester (which is now specified solely by $\calD$) has {\em soundness $\delta$} if for every $r\in \F_2^n$, $$\Pr_{\alpha\leftarrow \calD}[\alpha\cdot r=1] \geq \delta \cdot d(r,\calC),$$ where $d(r,\calC)=\min_{c\in \calC} d(r,c)$.  This formulation of soundness is often referred to as {\em strong soundness} in the literature.  
Weaker formulations of soundness in the literature only require that the test reject with good probability when $r$ is sufficiently far from the code.  Typically, 
we want $\delta=\Omega(1/d)$, where $d$ is the minimum distance of the code, so that received words at distance $\Omega(d)$ 
from $\calC$ are rejected with constant probability.  If there are received words at distance $\omega(d)$ from $\calC$, then it is common to cap the 
rejection probability at a constant (e.g. require
$\Pr_{\alpha\leftarrow \calD}[\alpha\cdot r=1] \geq \min\{\delta \cdot
d(r,\calC),1/3\}$), but we ignore this issue in the introduction for simplicity.

We are interested in two parameter regimes for LTCs:
\begin{description}
\item[Asymptotically Good LTCs]  Here we seek rate $k=\Omega(n)$,
  minimum distance $d=\Omega(n)$, soundness
  $\delta=\Omega(1/d)=\Omega(1/n)$, and query complexity $q=O(1)$.   Unfortunately, we do not know whether
such codes exist --- this is the major open problem about LTCs first posed by \cite{GoldreichSu06} , and it
is closely related to the long-standing open question about whether
SAT has constant-query PCPs of linear length (which would enable
proving that various approximation problems require time
$2^{\Omega(n)}$ under the exponential-time hypothesis).  The closest
we have is Dinur's construction~\cite{Dinur07}, which has
inverse-polylogarithmic (rather than constant) relative rate
(i.e. $n=k\cdot\polylog(k)$). A recent result by Viderman additionally
achieves strong soundness \cite{Viderman13}.
    
\item[Constant-distance LTCs]   Here we are interested in codes where the minimum distance $d$ is a fixed constant, and the traditional coding question is
how large the rate $k$ can be as $n\rightarrow \infty$.  BCH codes
have (optimal) rate $k=n-(d/2)\cdot \log n$, but do not have any local
testability properties.  \snote{or are they ``not {\em known} to have
  any local testability properties''}\pnote{By BHR, they cannot be
  local, since their dual distance
  is close to 1/2. Basically any local view (reading n/4 bits) looks like uniformly random bits.}
Reed--Muller codes yield the best known locally testable codes in this regime, with rate $k=n-O((\log n)^{\log d})$ 
and query complexity $q=O(n/d)$ to achieve soundness
$\delta=\Omega(1/d)$~\cite{BKSSZ}.\footnote{This is a case where the
  rejection probabilities need to be capped at a constant, since there
  may be received words at distance $\omega(d)$ from $\calC$.} (This query complexity is optimal, as $\Omega(n/d)$ queries is needed to detect $d/2$ random corruptions to a codeword with constant probability.)  
\snote{is this rate bound accurate as stated, or is more slop needed
  than the $O(\cdot)$, eg should the $O(\cdot)$ go in the exponent?}
\pnote{I think it is good as stated: for degree $m-r$ in $m$ vars, $n
  =2^m$, $n -k = {m \choose \leq r} =O(m^r), d = 2^r$.} An open problem is whether rate
$k=n-c_d\cdot \log n$ is possible, for some constant $c_d$ depending
only on $d$ (but not on $n$).
\end{description}

In terms of limitations of LTCs, there are a number of results
shedding light on the structure of an LTC with good parameters (see
the survey \cite{BenSasson10}), but there are essentially no
nontrivial upper-bounds on rate known for arbitrary $\F_2$-linear
LTCs. 

Instead of bounding the query complexity of our LTCs, it is convenient for us to work with {\em smooth LTCs}, where we simply require that the 
tester does not query any one coordinate too often.  Formally, a tester, specified by a distribution $\calD$ on $\calC^\perp$, is {\em $\eps$-smooth} if for
every $i\in [n]$, $\Pr_{\alpha\gets \calD}[\alpha_i=1] \leq \eps$.    This is analogous to the notion of smooth locally {\em decodable} codes (LDCs) defined by
Katz and Trevisan~\cite{KatzTr00}.  Like in the case of LDCs, bounding smoothness is almost equivalent to bounding query complexity, where query complexity $q$ 
corresponds to smoothness $\Theta(q/n)$ (as would be the case for testers that make $q$ uniformly distributed queries).  \footnote{An arbitrary $q$-query LTC can be converted into one that is $\eps$-smooth by discarding coordinates that are queried with probability more than $\eps$ (and treating them as zero in the tester); the only cost of this transformation is that the distance of the code may decrease by $q/\eps$, so we can set $\eps=O(q/d)$ and lose only a small constant factor in the distance.  Conversely, an $\eps$-smooth tester queries at most $\eps n$ coordinates in expectation, so we can get a tester with query complexity $q=O(\eps n)$ by discarding tests that query more than $q$ coordinates.  However, the latter transformation costs an {\em additive} constant (namely $\eps n/q$) in the soundness probability, and thus does not preserve {\em strong} soundness.}  In particular, we want the smoothness to be $\eps=O(1/n)$ in the asymptotically good regime, and $\eps=O(1/d)$ in the constant-distance regime.

\subsection{Cayley Graphs}

Cayley graphs are combinatorial structures associated with finite groups and are useful for applications ranging from pure group theory 
to reasoning about the mixing rates of Markov chains to explicit constructions of expander graphs \cite{Biggs,HLW}. 
In this paper, we focus on the case that the group is
a finite-dimensional vector space $\calV$ over $\F_2$.  (So $\calV\cong \F_2^h$ for some $h\in \N$.)  Given a multiset $S\subseteq \calV$, the Cayley (multi)graph $\Cay(\calV,S)$ has vertex set $\calV$ and edges $(x,x+s)$ for every $s\in S$ (with appropriate multiplicities if $S$ is a multiset).  Note that this is an $|S|$-regular undirected graph, since every element of $\calV$ is its own additive inverse.  
If we take $S$ to be a basis of $\calV$, then $\Cay(\calV,S)$ is simply the $h$-dimensional hypercube, where $h=\dim \calV$.   We will be interested in the properties of such graphs when $|S|$ is larger than $h$.

\subsection{LTCs and Metric Embeddings of Cayley Graphs}
\label{sec:metric-intro}

Our first result shows that locally testable codes are equivalent to Cayley graphs with
low-distortion metric embeddings into $\ell_1$. We refer the reader to \cite[Chapter 15]{Matousek} for background on metric embeddings.
An {\em embedding} of a metric space $(X_1,d_1)$ into metric space $(X_2,d_2)$ with {\em distortion} $c\geq 1$ is a function $f : X_1\rightarrow X_2$ such that for some $\alpha\in \R^+$ and every $x,y\in X_1$, we have 
$$\alpha \cdot d_1(x,y) \leq d_2(f(x),f(y)) \leq c\cdot \alpha \cdot d_1(x,y).$$
A commonly studied case is when $(X_1,d_1)$ is the shortest-path
metric $d_\calG$ on a graph $\calG$, and $(X_2,d_2)$ is an $\ell_p$
metric.  Indeed, we will take $(X_1,d_1)$ to be the shortest-path
metric on a Cayley graph, and $(X_2,d_2)$ to be an $\ell_1$ metric.
We will use the well-known characterization of $\ell_1$
metrics as the cone of the ``cut metrics'': a finite metric space
$(X_2,d_2)$ is an $\ell_1$ metric if and only if there is a constant
$\alpha\in \R^+$ and a distribution $F$ on boolean functions on $X_2$
such that for all $x,y\in X_2$, $d_2(x,y) = \alpha\cdot \Pr_{f\gets
  F}[f(x)\neq f(y)]$.  

Note that the shortest-path metric on the hypercube
$\Cay(\F_2^h,\{e_1,\ldots,e_h\})$ is an $\ell_1$ metric.  Indeed, this
metric is simply the Hamming distance $d$, and $d(x,y) = n \cdot
\Pr_{i\gets [n]}[x_i\neq y_i]$.    We show that the existence of
locally testable codes is equivalent to being able to approximate this
property (i.e. have low-distortion embeddings into $\ell_1$) even when
the number of generators is noticeably larger than $h$.  To avoid
trivial ways of increasing the number of generators (like duplicating
generators, or taking small linear combinations), we will also require
that the generators are $d$-wise linearly independent (i.e. have no
linear dependency of length smaller than $d$). 

\begin{Thm} \label{thm:metric-intro}
\begin{enumerate}
\item If there is an $\F_2$-linear code of blocklength $n$, rate $k$, and distance $d$ with an $\eps$-smooth local tester of soundness $\delta$, then there is a
Cayley graph $\calG=\Cay(\F_2^h,S)$ such that $|S|=n$, $h=n-k$, $S$ is
$d$-wise linearly independent, and the shortest path metric on $\calG$
embeds into $\ell_1$ with distortion at most
$\eps/\delta$. \label{itm:LTCtoEmbed} 

\item If there is a Cayley graph $\calG=\Cay(\F_2^h,S)$ such that
  $|S|=n$, $S$ is $d$-wise linearly independent, and the shortest path
  metric on $\calG$ embeds into $\ell_1$ with distortion at most $c$,
  then there is an $\F_2$-linear code of blocklength $n$, rate $k = n-h$, and distance $d$
with an $\eps$-smooth local tester of soundness $\delta$, for some
$\delta$ and $\eps$ such that $\eps/\delta\leq
c$. \label{itm:EmbedToLTC} 
\end{enumerate}
\end{Thm}

Note that the theorem provides an exact equivalence between locally
testable codes and $\ell_1$ embeddings of Cayley graphs, except that
the equivalence only preserves the ratio $\eps/\delta$ rather than the two quantities
separately. It turns out that this ratio is the appropriate parameter
to measure  when considering {\em strong}
soundness.  (See Section~\ref{sec:boosting}.)  For weaker notions of
soundness, we obtain equivalences with weaker notions of
low-distortion embeddings, such as ``single-scale embeddings'' (where we replace the requirement that
$d_2(f(x),f(y))\geq \alpha d_1(x,y)$ with $d_1(x,y)\geq D\Rightarrow
d_2(f(x),f(y))\geq \alpha D$, see \cite{Lee05} and references
therein).  


The theorem specializes as follows for the two main parameter regimes of interest:
\begin{Cor}
There is an asymptotically good smooth LTC with strong soundness ($k=\Omega(n)$, $d=\Omega(n)$, $\eps/\delta=O(1)$) iff there is a Cayley graph $\calG=\Cay(\F_2^h,S)$ with
$|S|=(1+\Omega(1)) h$ such that $S$ is $\Omega(h)$-wise linearly independent and the shortest path metric on $\calG$ embeds into $\ell_1$ with distortion $O(1)$.
\end{Cor}

\begin{Cor}
For a constant $d$, there is a distance $d$ LTC of blocklength $n$ with rate $k=n-c_d\log n$ and $\eps/\delta=O(1)$ iff there is a Cayley graph $\calG=\Cay(\F_2^h,S)$
with $|S|=2^{h/c_d}$ such that $S$ is $d$-wise linearly independent and the shortest path metric on $\calG$ embeds into $\ell_1$ with distortion $O(1)$.
\end{Cor}

To interpret the theorem, let's consider what the conditions on the Cayley graph $\calG$ mean.  The condition that $|S|=n$ and the elements of $S$ are
$d$-wise independent means that locally, in balls of radius $d$, the
graph $\calG$ looks like the $n$-dimensional hypercube (which embeds
into $\ell_1$ with no distortion).  However, it is squeezed into  a
hypercube of significantly lower dimension $h$ (which may make even
constant distortion impossible). 

The canonical example of graphs that do not embed well into $\ell_1$ are expanders.  Specifically, an $n$-regular expander on $H=2^h$ vertices with all nontrivial eigenvalues bounded away
from 1 requires distortion $\Omega(h/\log n)$ to embed into $\ell_1$.  Roughly speaking, the reason is that by Cheeger's Inequality (or the Expander Mixing Lemma),
cuts cannot distinguish random neighbors in the graph from random and independent pairs of vertices in the graph, and random pairs of vertices 
are typically at distance $\Omega(h/\log n)$.\footnote{Actually, for Cayley graphs over $\F_2$ vector spaces, the bound can be improved to
$\Omega(h/\log(n/h))$, using the fact that there are at most ${n \choose t}$ (rather than $n^t$) vertices at distance $t$ from any given vertex.}

Thus, saying that a graph $\calG$ embeds into $\ell_1$ with constant distortion intuitively means that $\calG$ is very far from being an expander.  
More precisely, to prove
the nonexistence of an $\ell_1$ embedding of distortion $c$ amounts to exhibiting a distribution $\Close$ on edges of $\calG$ and a distribution $\Far$
on pairs of vertices in $\calG$ such that for every cut $f : \F_2^h\rightarrow \zo$, 
$$\frac{\Pr_{(x,y)\gets \Close}[ f(x)\neq f(y)]}{c} > \frac{\Pr_{(x,y)\gets \Far}[f(x)\neq f(y)]}{\E_{(x,y)\gets \Far}[d_\calG(x,y)]}.$$
As discussed above, if $\calG$ were an expander, we could take $\Close$ to be the uniform distribution on edges and $\Far$ to be the uniform distribution on pairs of vertices,
and deduce a superconstant lower bound on $c$.   Showing an
impossibility result for LTCs amounts to finding
such expander-like distributions $\Close$ and $\Far$ in an arbitrary Cayley graph with a large (size $n$) set of
$d$-wise linearly independent generators $S$. 

Our construction of a Cayley graph from an LTC in
Item~\ref{itm:LTCtoEmbed} of Theorem~\ref{thm:metric-intro} is a
``quotient of hypercube''  construction previously analyzed by Khot
and Naor~\cite{KhotNaor}.  Specifically, they showed that if we start
from a code $\calC$ whose dual  code $\calC^\perp$ has large minimum
distance, then the resulting Cayley graph $\calG$ requires large
distortion to embed into $\ell_1$. Our contributions are to show that we can replace the 
hypothesis with the weaker condition that $\calC$ is not locally
testable, and to establish a tight converse by constructing LTCs from Cayley graphs with low-distortion embeddings.

\subsection{LTCs and Spectral Properties of Cayley Graphs.}  In our
second result, we show that locally testable codes are equivalent to
Cayley graphs with spectral properties similar to the ``$\eps$-noisy
hypercube''. We call such graphs derandomized hypercubes. 

For Cayley graphs over $\F_2$ vector spaces (and more generally abelian groups), the spectrum can be described quite precisely using Fourier analysis.  Let $M$ be the transition matrix for the random walk on $\Cay(\calV,S)$, i.e. the adjacency matrix divided by $|S|$.  Then, regardless of the choice of $S$, the eigenvectors of $M$ are exactly of the form $\chi_b(x) = (-1)^{b(x)}$ where $b : \calV\rightarrow \F_2$ ranges over all $\F_2$-linear functions.  (If we pick a basis so that $\calV=\F_2^h$, then each such linear function is of the form $b(x) = \sum_i b_ix_i$.)  The 
eigenvalue of $M$ associated with $\chi_b$ is 
$(1/|S|)\cdot \sum_{s\in S} \chi_b(s)$. In particular, if $S$ is a
$\lambda$-biased space~\cite{NaorNa93} for $\lambda$ bounded away from
1, then all the nontrivial eigenvalues have magnitude at most
$\lambda$, and hence the graph $\Cay(\calV,S)$ is an expander.  In
contrast, for the case of the hypercube ($S=\{e_1,\ldots,e_h\}$ for a
basis $e_1,\ldots,e_h$ of $\calV$), the eigenvalue associated with
$b=(b_1,\ldots,b_h)$ is $1-2|b|/h$ where $|b|$ is the Hamming weight
of $b$, so there are ${h \choose i}$ eigenvalues of value $1-2i/h$.   

In this section, it will be useful to generalize the notion of Cayley
graph from multisets to distributions over $\calV$.   If $S$ is a
distribution over $\calV$, then $\Cay(\calV,S)$ is a weighted graph
where we put weight $\Pr[S=s]$ on the edge $(x,x+s)$ for every $x,s\in
\F_2^h$. $\Pr[S=s]$ is also the $(x,x+s)$ entry of the transition
matrix of the random walk on $\Cay(\calV,S)$.  Now, the eigenvalues
are $\lambda(b) = \E_{s\gets S}[\chi_b(s)]$.  Here a useful example is
the $\eps$-noisy hypercube, where $\calV = \F_2^h$ and
$S=(S_1,\ldots,S_h)$ has each coordinate independently set to 1 with
probability $\eps$, and hence the eigenvalues are $\lambda(b) =
(1-2\eps)^{|b|}$. 

Neither the hypercube nor the $\eps$-noisy hypercube are very good expanders, as they have eigenvalues of $1-2/h$ and $1-2\eps$, respectively, corresponding 
to eigenvectors $\chi_b$ with $|b|=1$ (which in turn correspond to the
``coordinate cuts,'' partitioning $\F_2^h$ into the sets $\{x :
x_i=1\}$ and $\{x : x_i=0\}$).  However, their spectral properties do
imply that small sets expand well.  Indeed, Kahn, Kalai, and
Linial~\cite{KahnKaLi88} showed that the indicator vectors of
``small'' sets in $\F_2^h$ are concentrated on the eigenvectors
$\chi_b$ where $|b|$ is large, and hence small sets expand well in
both the hypercube and $\eps$-noisy hypercube (where ``small'' is
$|\calV|^{1-\Omega(1)}$ in the case of the hypercube, and $o(|\calV|)$
in the case of the $\eps$-noisy hypercube).

Our spectral characterization of locally testable codes is as follows.

\begin{Thm} \label{thm:spectrum-intro}
There is an $\F_2$-linear code of blocklength $n$, rate $k$, and distance $d$ with an $\eps$-smooth local tester of soundness $\delta$
if and only if there is a
Cayley graph $\calG=\Cay(\F_2^h,\calD)$ (for some distribution $\calD$ on $\F_2^h$) and a set $S=\{b_1,\ldots,b_n\}$ of linear maps
$b_i : \F_2^h\rightarrow \F_2$ satisfying:
\begin{enumerate}
\item $h=n-k$
\item $S$ is $d$-wise linearly independent
\item $\lambda(b_i)\geq 1-2\eps$ for $i=1,\ldots,n$.
\item For every linear map $b : \F_2^h\rightarrow \F_2$, $\lambda(b)
  \leq 1-2\delta\cdot \rk_S(b)$, where $\rk_S(b) = \min\{|T| :
  T\subseteq S, b=\sum_{i\in T} b_i\}$.\snote{changed $d_S$ to $\rk_S$
    for consistency with the text, though I don't quite see why we
    should think of this as a kind of ``rank''}\pnote{I called it rank
    in analogy with the rank of a quadratic form. Generators (products
    of linear forms) have rank 1, and for others it is the smallest
    $k$so that it is the sum of k generators. But there I think it
    corresponds to the rank of some matrix derived from the quadratic, so maybe not
    the best analogy. Using $d_S$ seems a bit strange also since
    $d_S(s) =1$ for all $s \in S$. Any better ideas? How about
    $\deg_S$? Its a macro, and
    easy to change}
\label{itm:spectral-soundness}
\end{enumerate}
\end{Thm}

Let's compare these properties with those of the $\eps$-noisy
hypercube.  Recall that, in the $\eps$-noisy hypercube, the coordinate
cuts $S=\{e_1,\ldots,e_h\}$ are linearly independent and all
give eigenvalues of $1-2\eps$.  And for every $b$, $\lambda(b) = (1-2\eps)^{|b|} = (1-2\eps)^{\rk_S(b)} = 1-\Omega(\eps\cdot \rk_S(b))$ (provided $\rk_S(b)\leq O(1/\eps)$). 

Like in our metric embedding result, the main difference here is that
we are asking for the set $S$ to be of size larger than $h$ (while
retaining $d$-wise independence among the generators), so we need to squeeze many large
eigenvalues into a low-dimensional space.  One reason that these
spectral properties are interesting is that they imply that the graph
$\calG$ is a small-set expander for sets of size $|\calV|/\exp(d)$(see Lemma \ref{lem:hypercon-sse}).   

One direction of the above theorem (from LTCs to Cayley graphs) is
extracted from the work of Barak et al.~\cite{BGH+12}, who used
locally testable codes (in the constant distance regime) to construct
small-set expanders that have a large number of large eigenvalues (as
a function of the number $H=2^h$ of vertices).  Such graphs provide
barriers to improving the analysis of the Arora--Barak--Steurer algorithm
for approximating small-set expansion and unique games~\cite{ABS10},
and were also used by Barak et al.~\cite{BGH+12} to construct improved
integrality gap instances for semidefinite programming relaxations of
the unique games problem.  Our contribution is showing that the
connection can be reversed, when formulated appropriately
(in terms of spectral properties rather than small-set expansion).

We can specialize Theorem~\ref{thm:spectrum-intro} to the two
parameter regimes of interest to us (see Corollaries
\ref{cor:sse-asymptotic} and \ref{cor:sse-constant} in Section
\ref{sec:sse-cons} for precise statements). The existence of
asymptotically good smooth LTCs with strong soundness is equivalent to
the existence of Cayley graphs whose eigenvalue spectrum resembles the
$n$-dimensional Boolean hypercube (for eigenevalues in the range $[0.5,1]$) but where the
number of vertices is $2^{(1- \Omega(1))n}$. In the constant $d$
regime, the existence of $[n, n -c_d\log n, d]_2$ LTCs blocklength $n$
with smoothness $\eps=O(1/d)$, and soundness $\delta=\Omega(1/d)$  is
equivalent to the existence of Cayley graphs whose eigenvalue spectrum
resembles the $n$-dimensional Boolean hypercube (for eigenvalues in the range $[0.5,1]$) but where the
number of vertices is $n^{c_d}$.  \snote{added corollaries, but I'm
  not sure how much they add for the reader at this point...  they are
  not quite as concise as the corollaries of the metric embedding
  result.}\pnote{Agreed. Moving them out to the suburbs, replaced with
  more vague but easy to read statements.}

\eat{
\begin{Cor}
There is an asymptotically good smooth LTC with strong soundness ($k=\Omega(n)$, $d=\Omega(n)$, $\eps=O(1/n)$, $\delta=\Omega(1/n)$) iff there is a Cayley graph $\calG=\Cay(\F_2^h,\calD)$ and a set $S$ of linear maps $b : \F_2^h\rightarrow \F_2$
such that 
\begin{enumerate}
\item $|S|=(1+\Omega(1))\cdot h$,
\item $S$ is $\Omega(h)$-wise linearly independent, 
\item $\lambda(b) \geq 1-O(1/h)$ for every $b\in S$, and 
\item For every
linear map $b : \F_2^h\rightarrow \F_2$, we have $\lambda(b) \leq 1-\Omega(\rk_S(b)/h)$.
\end{enumerate}
\end{Cor}

\begin{Cor}
For a constant $d$, there is a distance $d$ LTC of blocklength $n$ with rate $k=n-c_d\log n$, smoothness $\eps=O(1/d)$, and soundness $\delta=\Omega(1/d)$ iff there is a Cayley graph $\calG=\Cay(\F_2^h,\calD)$ and a set $S$ of linear maps $b : \F_2^h\rightarrow \F_2$ such that
\begin{enumerate}
\item $|S|=2^{h/c_d}$,
\item $S$ is $d$-wise linearly independent, 
\item $\lambda(b) \geq 1-O(1/d)$ for every $b\in S$, and 
\item For every
linear map $b : \F_2^h\rightarrow \F_2$, we have $\lambda(b) \leq 1-\Omega(\rk_S(b)/d)$.
\end{enumerate}
\end{Cor}}

Like our metric embedding result, Theorem~\ref{thm:spectrum-intro} and its corollaries have analogues for weaker notions of soundness for the locally testable codes.  Specifically,
Item~\ref{itm:spectral-soundness} changes in a way that is analogous
to the soundness condition, for example only requiring that
$\lambda(b)$ is small when $\rk_S(b)$ is large.

\subsection{Perspective}

For many of the problems about constructing codes or Cayley graphs
studied in theoretical computer science, the main challenge is finding
an {\em explicit} construction.  Indeed, we know that a randomly
chosen code has good rate and distance and that a randomly chosen set
of generators yields a Cayley graph with high expansion, and much of
the research on these topics is aimed at trying to match these
parameters with efficient deterministic algorithms. 

Locally testable codes (and the equivalent types of Cayley graphs that
we formulate) are intriguing in that they combine properties of random
objects (such as large distance) with very non-random properties (the
existence of a local tester).  Thus the major open questions (such as
whether there are asymptotically good LTCs) are {\em existential} ---
do there even exist objects with the given parameters, regardless of
the complexity of constructing them?   

Our hope is that the alternative characterizations developed in this
paper will be useful in approaching some of these existential
questions, either positively (e.g. by using graph operations to
construct Cayley graphs with the properties discussed above,
analogously to Meir's construction of LTCs~\cite{Meir09}) or
negatively (e.g. by reasoning about expander-like subgraphs of Cayley
graphs, as discussed above in Section~\ref{sec:metric-intro}).  

\snote{say something about applications?}
\pnote{New para:}
The connection between metric embeddings and local testability gives a
new perspective on existing results in this area, for instance we use
it to give a simple LTC-based proof of the non-embeddability result of
Khot and Naor \cite{KhotNaor} (see Section
\ref{sec:reform}). Similarly, the connection to derandomized
hypercubes has been used by ~\cite{BGH+12,KaneMeka} to construct improved
integrality gap instances for semidefinite programming relaxations of
combinatorial optimization problems.

\section{Locally Testable Codes revisited}

In this section we reformulate the properties
of \LTC s in terms of cosets,  which makes our equivalences easier to show.

Recall that a local tester for an $[n,k,d]_2$ binary linear code
$\calC$ is specified by a a distribution $\calD$ on $\calC^\perp$.
The tester $\calD$ is $\eps$-smooth if for
every $i \in [n]$, $\Pr_{\alpha \samp \calD}[\alpha_i =1] \leq \eps$.
For $v \in \F_2^n$ let  $d(v,\calC) = \min_{c \in \calC}d(v,c)$ and $\Rej(v,\calD)  = \Pr_{\alpha \samp \calD}[\alpha \cdot v =1]$. 
We say that $\calD$ has soundness $\delta$ if $\Rej(v,\calD) \geq
\delta d(v,\calC)$ for all $v \in \F_2^n$. We say that an $[n,k,d]_2$ linear code is $(\eps,\delta)$-locally
testable if it has a tester $\calD$ which has smoothness $\eps$ and
soundness $\delta$.  By considering received words at distance
$1$ from the code,  we get $\delta \leq \Rej(e_i,\calD) \leq \eps$.
The upper bound is an easy consequence of the smoothness. Ideally, we want $\delta = \Omega(\eps)$.

Given $v \in \calV$, let $\bar{v} \in \cosC$ denote the coset
of $\calC$ containing it. Let $\bar{\calE} =
\{\bar{e}_1,\ldots,\bar{e}_n\}$ denote the coset representatives of
the basis vectors $\{e_1,\ldots,e_n\}$. The $\bar{e_i}$s are not independent over $\F_2$, indeed we have
\begin{align*}
\sum_{i \in S}\bar{e}_i = 0 \iff \sum_{i\in S} e_i \in \calC
\end{align*}
Hence the shortest non-trivial linear dependence is of weight exactly
$d$.

For $\bar{v} \in \cosC$, there could be several ways to write
it as a linear combination over $\bcalE$. We have
\begin{align*}
\bar{v} = \sum_{i \in S}\bar{e}_i\iff v + \sum_{i\in S} e_i \in \calC
\end{align*}
Hence if we define $d(\bar{v},\calC) = d(v,\calC)$ for any $v \in
\bar{v}$ (the exact choice does not matter), it follows that
$d(\bar{v},\calC) = \rk_{\bar{\calE}}(v)$. Similarly, for every  $c \in \calC$, $v \in \calV$ and $\alpha \in
\calC^\perp$, $\alpha(v) = \alpha(v +c)$. Hence
\begin{align}
\label{eq:rej-coset}
\Rej(v,\calD) = \Rej(v + c,\calD)
\end{align}
This lets us define $\Rej(\bar{v},\calD) = \Rej(v,\calD)$ for any $v \in
\bar{v}$.

We can now rephrase smoothness and soundness in terms of
coset representatives.
\begin{align}
\label{eq:smooth-rej}
\Pr_{\alpha \samp \calD}[\alpha_i =1] = \Pr_{\alpha \samp
  \calD}[\alpha (e_i) =1] = \Rej(\bar{e_i},\calD)
\end{align}
Thus $\calD$ is $\eps$-smooth if every $\bar{e_i}$ is rejected with probability at most $\eps$.

We say a set $S$  of vectors in an $\F_2$-linear space is $d$-wise independent
if every $T \subseteq S$ where $|T| < d$ is linearly independent over $\F_2$. 
For a set of vectors $S = \{s_1,\ldots,s_n\}$ which span a space
$\calT$, we use $\rk_S(t)$ for $t \in \calT$ to denote the smallest $k$
such that $t$ can be expressed as the sum of $k$ vectors
from $S$.  With this notation, $\calD$ has soundness $\delta$ if for every $\bar{v} \in
\cosC$ such that $\rk_{\bar{\calE}} \geq d'$, $\Rej(\bar{v}) \geq
\delta d'$.

We summarize these observations in the following lemma:

\begin{Lem}
\label{lem:ltc-coset}
Let $\calC$ be an $[n,k]_2$ code and let $\calD$ be a tester for $\calC$.
\begin{itemize}
\item $\calC$ has distance $d$ iff the set $\bcalE$ is $d$-wise
  independent.
\item The tester $\calD$ is $\eps$-smooth iff $\Rej(\bar{e_i},\calD)
  \leq \eps$ for all $i \in [n]$.
\item For $\bar{v} \in \cosC$, $d(\bar{v},\calC) = \rk_{\bar{\calE}}(\bar{v})$. Hence $\calD$ has soundness $\delta$ iff for every $\bar{v} \in
\cosC$,
$$\Rej(\bar{v},\calD) \geq \delta \cdot \rk_\bcalE(\bar{v})$$
\end{itemize}
\end{Lem}

\section{Locally Testable Codes and Metric Embeddings}

Let $S =\{s_1,\ldots,s_n\}  \subset \F_2^h$ be set of $n \geq h$ generators of $\F_2^h$
that are $d$-wise independent. Let $\calG = \Cay(\F_2^h,S)$ be the Cayley
graph whose edges correspond to the set $S$. The graph $\calG$ can be naturally associated with a
code $\calC_\calG$ which consists of all vectors $c =
(c_1,\ldots,c_n)$ such that $\sum_ic_is_i =0$. It is easy to see that
$\calC_\calG$ is an $[n,n-h,d]_2$ linear code. 

Similarly, one can start from an $[n,k,d]_2$ linear code $\calC$ and
construct a Cayley graph $\calG_\calC$ on $\F_2^{n-k}$.  We take the
vertex set to be $\F_2^n/\calC$. We add an edge $(\bar{x},\bar{y})$ if
there exist $x \in \bar{x}$ and $y \in \bar{y}$ such that $d(x,y) = 1$. It is easy to see that is equivalent to taking $S =
\{\bar{e_1},\ldots,\bar{e_n}\}$, and this set is $d$-wise
independent by the distance property of $\calC$. 

It is easy to see that this construction inverts the
previous construction. Henceforth we will fix a code  $\calC$
and a graph  $\calG$ that can be derived from one another. The vertex
set of $\calG$ is given by $V(\calG)=  \F_2^n/\calC$ and the edge set
$E(\calG)$ by $\{\bar{x},\bar{x} + \bar{e_i}\}$ for $\bar{x} \in
\F_2^n/\calC$ and $i\in [n]$.  
 
\begin{Lem}
Let $d_\calG$ denote the shortest path metric on $\calG$. We have
\[ d_\calG(\bx,\by) = d_\calG(\bx + \by,0) = d(x + y,\calC)\]
\end{Lem}
\begin{Proof}
If $d_\calG(\bar{x},\bar{y}) \leq d$ then there exists $T \subset
[n]$ of size $d$ such that 
\[\bar{y} = \bar{x} + \sum_{j \in T}\bar{e}_j \ \ \Rightarrow \ \ \bar{x} + \bar{y} = \sum_{j \in T}\bar{e}_j \]
hence $d(\bar{x} + \bar{y},\calC) \leq d$. Similarly, if
$d(x + y,\calC) \leq d$, that gives an $\bar{x}$--$\bar{y}$ path of length $d$ in
$\calG$.
\end{Proof}

An $\ell_1$-embedding of the shortest path metric $d_\calG$ on the graph
$\calG$ is a distribution $\calD$ over Boolean functions $f:V(\calG) =
\F_2^n/\calC \rgta \pmo$. The distance $\delta(\bx,\by)$ between
a pair of vertices $\bx$ and $\by$ under this embedding is 
\begin{align*}
\delta(\bx,\by) = \Pr_{f \in \calD}[f(\bx) \neq f(\by)]
\end{align*} 
We define the stretch of an edge $(\bx,\by)$ to be the ratio
$\delta(\bx,\by)/d(\bx,\by)$.  The distortion $c_\calD$ of the
embedding $\calD$ is the ratio of the maximum to the minimum stretch
of any pair of vertices. It is given by  
\begin{align*}
c_\calD & = \frac{\max_{\bx,\by \in V(\calG)}\delta(\bx,\by)/d_\calG(\bx,\by)}{\min_{\bx,\by}\delta(\bx,\by)/d_\calG(\bx,\by)}
\end{align*}
It follows by the triangle inequality that the stretch is maximized by some edge. Hence we get
\begin{align}
\label{eq:distortion}
c_\calD = \frac{\max_{\bx\in
    V(\calG),i\in[n]}\delta(\bx,\bx+\be_i)}{\min_{x,y}\delta(\bx,\by)/d_\calG(\bx,\by)}  
\end{align}
The minimum $c$ achieved over all $\ell_1$-embeddings of $\calG$ is
denoted $c_1(\calG)$.

\begin{Def}
An embedding $\calD$ of $\calG =\Cay(\F_2^h,S)$ into $\ell_1$ is linear
if $\calD$ is supported on functions $\chi_\alpha(x) =
(-1)^{\alpha(x)}$ where $\alpha$ is an $\F_2$-linear function on $V(\calG)$.
\end{Def}

The space of linear functions on $\F_2^n/\calC$ is
isomorphic to $\calC^\perp$.  In a linear embedding, we have 
\begin{align*}
\delta(\bx,\by) & = \Pr_{\alpha \in \calD}[\chi_\alpha(\bx) \neq \chi_\alpha(\by)]\\
& = \Pr_{\alpha \in \calD}[\chi_\alpha(\bx + \by) \neq 1]\\
& = \delta(\bx + \by,0)
\end{align*}
Thus, the distance $\delta$ is invariant under shifting in linear embeddings,
just like the shortest path distance $d_\calG$. Indeed, the next lemma
shows that we can replace any embedding by a linear embedding without
increasing the distortion.

\begin{Lem}
\label{lem:linear}
There is a linear embedding $\calD$ of $\calG$ into $\ell_1$
achieving distortion $c_1(\calG)$.
\end{Lem}

To prove this lemma, we set up some machinery. Let $f:\F_2^n/
\calC \rgta \pmo$ be a Boolean function on $\F_2^n/\calC$. We can
extend $f$ to a function on all of $\F_2^n$ by setting $f(x) =
f(\bar{x})$. (We will henceforth switch freely between both notions). 
The resulting function is invariant under cosets of
$\calC$, which implies that its Fourier spectrum is supported on
$\calC^\perp$. Hence we have
\begin{align*}
f(x) = \sum_{\alpha \in \calC^\perp} \hat{f}(\alpha)\chi_\alpha(x)
\end{align*}
Further, we have $\sum_{\alpha \in \calC^\perp} \hat{f}(\alpha)^2
=1$. 

We now proceed to the proof of Lemma \ref{lem:linear}. 
\begin{Proof}
Given an arbitrary distribution $\calD$ on Boolean functions, we define a new distribution
$\calD'$ on functions where we sample $a \in \F_2^n$ uniformly at
random, $f \in \calD$, and return the function $f':\F_2^n \rgta \pmo$
defined by $f'(x) = f(x +a)$. We will show that $c_\calD' \leq c_\calD$. 

Let $\delta'(x,y)  = \Pr_{f' \in \calD'}[f'(x) \neq
  f'(y)]$. We have
\begin{align*}
\delta'(x,y) & = \Pr_{a \in \F_2^n, f \in \calD}[f(x +a) \neq f(y+a)]\\ 
& = \E_{a \in  \F_2^n}\left[\Pr_{f \in \calD'}[f(x +a) \neq
    f(y+a)]\right]\\
& = \E_{a \in  \F_2^n}[\delta(x +a, y+a)].
\end{align*}
Let $a_1$ and $a_2$ be the values of $a$ that minimize and maximize
$\delta(x+a,y+a)$ respectively. Then
\begin{align*}
\delta(x +a_1,y + a_1) \leq \delta'(x,y) \leq \delta(x +a_2,y + a_2)
\end{align*}
Hence we have
\begin{align*}
\frac{\delta(x +a_1,y + a_1)}{d(x+a_1,y+a_1)} \leq \frac{\delta'(x,y)}{d(x,y)} \leq \frac{\delta(x +a_2,y + a_2)}{d(x+a_2,y+a_2)}
\end{align*}
since all the denominators are equal.
But this implies that
\begin{align*}
\min_{x,y}\frac{\delta'(x,y)}{d(x,y)} \geq
\min_{x,y}\frac{\delta(x,y)}{d(x,y)},\\
\max_{x,y}\frac{\delta'(x,y)}{d(x,y)} \leq
\max_{x,y}\frac{\delta(x,y)}{d(x,y)}
\end{align*}
and hence
\begin{align*}
c_{\calD'} =
\frac{\max_{x,y}\frac{\delta'(x,y)}{d(x,y)}}{\min_{x,y}\frac{\delta'(x,y)}{d(x,y)}}
\leq
\frac{\max_{x,y}\frac{\delta(x,y)}{d(x,y)}}{\min_{x,y}\frac{\delta(x,y)}{d(x,y)}} =
c_{\calD}.
\end{align*}

Next we show that there is a linear embedding $\calD''$ with
distortion $c_{\calD''} = c_{\calD'}$. The embedding is simple to describe: we first
sample $f \in \calD$, we then sample $\chi_\alpha \in \calC^\perp$ with
probability $\hat{f}(\alpha)^2$. We denote this distribution on
$\calC^\perp$ by $\hat{f}^2$. Note that
\begin{align*}
\delta'(x,y) &  = \Pr_{a \in \F_2^n, f \in \calD}[f(x +a) \neq f(y+a)]\\
& = \E_{f \in \calD}\left[\frac{1}{2}\E_{a\in \F_2^n}[1 - f(x +a)f(y+a)]\right]\\
& = \E_{f \in \calD}\left[\frac{1}{2}\E_{a\in \F_2^n}\left[1 -
  \left(\sum_{\alpha \in \calC^\perp}\hat{f}(\alpha)\chi_\alpha(x
  +a)\right)\left(\sum_{\beta \in \calC^\perp}\hat{f}(\beta)\chi_\beta(y
  +a)\right)\right]\right]\\
& = \E_{f \in \calD}\left[\frac{1}{2}\left(1 -
  \sum_{\alpha,\beta \in
    \calC^\perp}\hat{f}(\alpha)\hat{f}(\beta)\chi_\alpha(x)\chi_\beta(y)\E_{a
    \in \F_2^n}[\chi_\alpha(a)\chi_\beta(a)]\right)\right]\\
& = \E_{f \in \calD}\left[\frac{1}{2}\left(1 - \sum_{\alpha
    \in \calC^\perp}\hat{f}(\alpha)^2\chi_\alpha(x)\chi_\alpha(y)\right)\right]\\
& = \E_{f \in \calD}\left[\sum_{\alpha \in \calC^\perp}\hat{f}(\alpha)^2\frac{(1 - \chi_\alpha(x)\chi_\alpha(y)}{2}\right]\\
& = \Pr_{f \in \calD}\Pr_{\alpha \in \hat{f}^2}[\chi_\alpha(x) \neq \chi_\alpha(y)]\\
& = \delta''(x,y).
\end{align*}
From this it follows that $c_\calD'' = c_\calD' \leq c_\calD$. 

The lemma follows by taking $\calD$ to be the $\ell_1$ embedding of $\calG$ that minimizes distortion.
\end{Proof}

We now prove the main result of this section.

\begin{Thm}
\label{thm:embedding}
We have $c_1(\calG) \leq c$ iff there exists an $(\eps,\delta)$-tester
for $\calC$ where $\delta \geq \eps/c$.
\end{Thm}
\begin{Proof}
For linear embeddings, we can use shift invariance to simplify the
expression for distortion. Since $\calD$ is a distribution on
$C^\perp$, we can view it as a tester for $\calC$. Note that
$\Rej(\bx,\calD) = \delta(\bx,0)$. Recall by Equation \eqref{eq:distortion}
\begin{align*}
c_\calD = \frac{\max_{\bx\in V(\calG),i\in[n]}\delta(\bx,\bx+\be_i)}{\min_{x,y}\delta(\bx,\by)/d_\calG(\bx,\by)}
\end{align*}
We can use $\delta(\bx,\by) = \delta(\bx +\by,0) = \Rej(\bx
+\by,\calD)$ to rewrite this as
\begin{align}
\label{eq:lin-distortion}
c_\calD = \frac{\max_{i\in[n]}\Rej(\be_i,\calD)}{\min_{\bx \in V(\calG)}\Rej(\bx,\calD)/d(\bx,0)}
\end{align}

Given a linear embedding specified by a distribution $\calD$ that
gives distortion $c_\calD$, we view $\calD$ as a tester. By definition, it has smoothness $\eps$ for 
\begin{align} 
\label{eq:smooth-tight}
\eps \geq \max_{i\in[n]}\Rej(\be_i,\calD)
\end{align} 
and has soundness $\delta$ for 
\begin{align}
\label{eq:sound-tight}
\delta \leq \min_{\bx \in V(\calG)}\Rej(\bx,\calD)/d(\bx,0)
\end{align}
since any such $\delta$ satisfies the condition
\[ \Rej(\bx,\calD) \geq \delta d(\bx,0).\]
By taking $\eps,\delta$ to satisfy Equations \ref{eq:smooth-tight} and
\ref{eq:sound-tight} with equality, we get $\delta = \eps/c_\calD$.

In the other direction, assume we have a $(\eps,\delta)$-tester for
$\calC$ where $\delta \geq \eps/c$. Note that $\eps, \delta$ must
satisfy Equations 
\ref{eq:smooth-tight} and \ref{eq:sound-tight}. Plugging these into
Equation \ref{eq:lin-distortion}, we get
\begin{align*}
c_\calD = \frac{\max_{i\in[n]}\Rej(\be_i,\calD)}{\min_{\bx \in
    V(\calG)}\Rej(\bx,\calD)/d(\bx,0)} \leq \frac{\eps}{\delta}
\leq c.
\end{align*}
\end{Proof}

\eat{
\begin{Cor}
\label{cor:metric-asymptotic}
There exists an asymptotically good family of codes $\{\calC_n\}$
where $\calC_n$ has blocklength $n$ and is $(O(1/n),\Omega(1/n))$-locally testable iff for infinitely many $h$ there
exists a Cayley graph $\calG_h = \Cay(\F_2^h,S)$ where
\begin{itemize}
\item $|S| \geq \rho_1 h$ for $\rho_1 > 1$, 
\item the elements of $S$ are $(\rho_o h)$-wise independent for $\rho_o > 0$,
\item $c_1(\calG) =O(1)$.
\end{itemize}
\end{Cor}

\begin{Cor}
\label{cor:metric-constant}
Let $d \geq 3$. There exists an asymptotic family of
codes $\{\calC_n\}$ where $\calC_n$ has parameters $[n,n - c_d\log
  n,d]_2$ and is $(\eps,\Omega(\eps))$-locally testable \footnote{for $\eps =
O(1/d)$. See Section \ref{sec:boosting} for a discussion of
  the possible values of $\eps$.}  iff for infinitely many $h$ 
there exists a Cayley graph $\calG_h = \Cay(\F_2^h,S)$ where
\begin{itemize}
\item $|S| = 2^{h/c_d}$ where $c_d \geq 1$ is a function of $d$,
\item the elements of $S$ are $d$-wise independent,
\item $c_1(\calG) =O(1)$.
\end{itemize}

\end{Cor}
}

\subsection{Boosting the soundness}
\label{sec:boosting}

Theorem \ref{thm:embedding} implies the existance of an
$(\eps,\delta)$-tester for $\calC$, where 
\[ \frac{\eps}{c_1(\calG)} \leq \delta \leq \eps. \] 
While this is the best ratio possible between $\eps$ and $\delta$,
Theorem \ref{thm:embedding} does not seem to guarantee the right
absolute values for them. In this section, we show that one can
achieve this by repeating the tests.  First, we identify the right
absolute values.

Let $t$ denote the covering radius of 
$\calC$, and let $\bx$ be a codeword at distance $t$ from
$\calC$. Since
\[1 \geq \Rej(\bx,\calD) \geq \delta t\]
we get $\delta \leq 1/t$. Given this upper bound, we would like
$\eps$ to be $\Theta(1/t)$ and $\delta$ to be
$\Theta(1/(c_1(\calG)t))$. We show that this is
possible, with a small loss in constants (which we do not attempt
to optimize).

\begin{Thm}
\label{thm:cov-radius}
There is a $(1/(4t), 1/(16c_1(G)t))$-tester for $\calC$.
\end{Thm}

We defer the proof of this result to Appendix \ref{app:boosting}.

Note that if $d =\Omega(n)$, then $d$ and $t$ differ by a constant
factor. However, when $d =o(n)$, it could be that $t =\omega(d)$.
In this case, we could relax the soundness requirement for words at distance $\omega(d)$ as follows:
\begin{align*} 
\Rej(\bx,\calD) \geq \begin{cases} \delta d(\bx,\calC) & \text{if}
\  d(\bx,\calC) \leq d \\
\delta d & \text{if} \ d(\bx,\calC) \geq d
\end{cases}
\end{align*}
It is possible to get such a tester where $\eps =
O(1/d)$, $\delta = \Omega(1/(c_1(\calG)d))$ using the same
argument as above, but replacing $t$ with $d$. We omit the details.

\subsection{Relation to previous work}
\label{sec:reform}

This equivalence allows us to reformulate results about LTCs in the
language of metric embeddings and vice versa. We present two examples
where we feel such reformulations are particularly
interesting.

\paragraph{Embedding lower bounds from dual distance: }

Khot and Naor show the following lower bound for $c_1(\calG)$ in terms
of its dual distance. 

\begin{Thm}
\label{thm:kn}
\cite[Theorem 3.4]{KhotNaor}
Let $\calC$ be an $[n,n-h]_2$ code and let $\calG$ be the associated
Cayley graph. Let $d^\perp$ denote its dual distance. Then
\begin{align*}
c_1(\calG) \geq \Omega\left(d^\perp \frac{h}{n\log(n/h)}\right)
\end{align*}
\end{Thm}

In the setting where $h/n = \Omega(1)$ (which is necessary to have
constant relative distance), this gives a lower bound of
$\Omega(d^\perp)$. Thus their result can be seen as the embedding
analogue of the result of BenSasson \etal\ \cite{BenSassonHaRa05}, who showed that the
existence of low-weight dual codewords is a necessary condition for
local testability. Our results allow for a simple alternative proof of Theorem \ref{thm:kn}.

\begin{proof}[Proof of Theorem~\ref{thm:kn}]
By Theorem \ref{thm:embedding}, there exists a $(\eps,\delta)$-tester $\calD$ so that
$c_1(\calG) = \epsilon/\delta$. We may assume without
loss of generality that $\calD$ is supported on non-zero code-words in
$\calC^\perp$, each of which is of weight at least $d^\perp$, so we
have $\eps\geq d^{\perp}/n$. 

As in Section \ref{sec:boosting}, we have $\delta \leq 1/t$ where $t$ is the covering radius of $\calC$.  
We lower bound $t$ by a standard volume argument:
\[ 2^{n-h} \cdot \sum_{t' = 0}^{t} {n \choose t'} \geq 2^n
\ \ \Rightarrow \ \ t = \Omega\left(\frac{h}{\log (n/h)}\right).\] 

So we have
\[ c_1(\calG) \geq \frac{\eps}{\delta} \geq \frac{d^{\perp}}{n}t  = \Omega\left(d^{\perp}
\frac{h}{n\log(n/h)}\right) . \]
\end{proof}

\paragraph{Lower bounds for basis testers: }

A basis tester for a code $\calC$ is a tester $\calD$ which is
supported on a basis for $\calC$. Ben-Sasson \etal\ showed a strong
lower bound for such testers \cite{BGKSV}. Their main result when
restated in our notation says:

\begin{Thm}
\cite[Theorem 5]{BGKSV}
Let $\calC$ be an $[n,k,d]_2$ code with an $(\eps,\delta)$-basis
tester. Then 
\[ \frac{\eps}{\delta} \geq \frac{k d}{3n}.\]
\end{Thm}

If $\calC$ is an $[n,k,d]_2$ code, a basis
tester for $\calC$ yields an embedding into $(n-k)$-dimensional
space. Hence their result implies that any linear embedding of
$\F_2^n/\calC$ into $(n-k)$-dimensional space requires distortion
$\Omega(kd/n)$ (even though $\F_2^n/\calC$ has dimension $n-k$ as a
vector space over $\F_2$), and hence low-distortion embeddings must
have larger support.  Note that since our reduction from arbitrary embeddings
to linear embeddings could blow up the support, this does
not imply a similar bound for arbitrary embeddings.

\eat{
Our results allow for an alternative derivation of Theorem
\ref{thm:kn} using the classical lower bound for embeddings using the
spectral gap (see for instance \cite[Chapter 15]{Matousek}. We outline the proof below:
\begin{enumerate}
\item The eigenvalue spectrum of $\calG$ is determined by the weight
  distribution of $\calC^\perp$. Each codeword $\alpha \in
  \calC^\perp$ corresponds to an eigenvalue $\lambda(\alpha) = 1
  -2\wt(\alpha)/n$. Thus the spectral gap is $2d^\perp/n$.
\item When $h = \Omega(n)$, it is easy to see that average pairwise
  distance of vertices in $\calG$ is $\Omega(n)$, since the number of
  vertices within a ball of radius $t$ around a vertex is bounded by
\[ {n \choose t} \ll 2^h \] 
for a suitable choice of $t = \Omega(n)$. 
\item If we now consider the uniform distribution on edges versus the
  uniform distribution on all pairs of vertices, by (1) we have
\[ \frac{\E_{(x,y) \in E(\calG)}[\delta(x,y)]}{\E_{(x,y) V(\calG) \times
    V(\calG)}[\delta(x,y)]} \geq \frac{2d^\perp}{n}\]
whereas by (2) we have
\[ \frac{\E_{(x,y) \in E(\calG)}[d(x,y)]}{\E_{(x,y) V(\calG) \times
    V(\calG)}[d(x,y)]} \leq O(1/n)\]
which implies that $c_1(\calG) \geq \Omega(d^\perp)$.
\end{enumerate}
}

\section{Locally Testable Codes and Derandomized Hypercubes}

\subsection{Derandomized Hypercubes}

We consider Cayley graphs over groups of characteristic $2$. Let
$\calA = \F_2^h$ for some $h > 0$.  A distribution $\calD$ over $\calA$ gives rise to a
weighted graph $\Cay(\calA,\calD)$ where the weight of edge $(\alpha,\beta)$
equals $\calD(\alpha + \beta)$. 

The symmetry of Cayley graphs makes it easy to compute their
eigenvectors and eigenvalues explicitly. Let $\calA^*$ denote the space of
all linear functions $b: \calA \rgta \F_2$. The characters of the group $\calA$
are in $1$-$1$ correspondence with linear functions: $b \in \calA^*$ corresponds to a
character $\chi_b: \calA \rgta \pmo$  given by $\chi_b(\alpha) = (-1)^{b(\alpha)}$. 
The eigenvectors of $\Cay(\calA,\calD)$ are precisely the characters $\{\chi_b\}_{b
  \in \calA^*}$. The corresponding eigenvalues are given by
$\lambda(b) = \E_{\alpha \in \calD}[\chi_b(\alpha)]$.

\eat{We view the set $\calA^*$ of linear functions as an $\F_2$ linear
space, the set can have linear dependencies  over $\F_2$. }

As mentioned earlier, the Cayley graphs we are interested in can be
viewed as derandomizations of the $\eps$-noisy hypercube, which retain
many of the nice spectral properties of the Boolean hypercube. Before
defining them formally, we list these properties that we would like
preserved (at least approximately). 

\begin{enumerate}
\item {\bf Large Eigenvalues.} There are $h$ ``top'' eigenvectors
  $\{\chi_{e_i}\}_{i=1}^h$ whose eigenvalues satisfy
  $\lambda(e_i) \geq 1- \eps$. 
\item {\bf Linear Independence.} The linear functions $\{e_1,\ldots,e_h\}$
  corresponding to the top eigenvectors are linearly independent over $\F_2$.
\item {\bf Spectral Decay.} For $a \in \calA^*$, if $a = \sum_{i \in
  S} e_i$, then $\lambda(a) \leq  (1 - \eps)^{|S|}$.
\end{enumerate}

\eat{
These properties are sufficient to guarantee $2-4$ hypercontractivity,
which gives small-set expansion and more. Hypercontractivity is at the heart of many deep
results in the analysis of Boolean functions, including the
Kahn-Kalai-Linial Theorem (KKL) and Freidgut's Junta Theorem.}

We are interested in Cayley graphs whose threshold rank $n$ is
possibly (much) larger than $h$. But this means that the corresponding dual vectors which
lie in the space $\calA^*$ of  dimension $h < n$ can no longer be linearly independent. So we
relax the Linear Independence condition, and only ask that there
should be no {\em short} linear dependencies between these
vectors. The Spectral Decay condition will stay the same, except that
we need to modify the notion of rank to account for linear
dependencies.

\begin{Def} \label{def:cayley-spectrum}
Let $\Cay(\calA,\calD)$ be a Cayley graph on the group $\calA = \F_2^h$.
Let $\mu, \nu \in [0,1]$ and $d \in \{1,\ldots,n\}$. 
Let $\calB^* = \{b_1,\ldots,b_n\}$ be a $d$-wise independent set
of generators for $\calA^*$ of cardinality $n$. We say that $\calB^*$ is a
$(\mu,\nu)$-spectrum generator for $\Cay(\calA,\calD)$ if it satisfies the following properties:
\begin{itemize}
\item {\bf Large Eigenvalues.} $\lambda(b) \geq 1 - \mu$ for every  $b \in \calB^*$.

\item {\bf Spectral Decay.} For $a \in \calA^*$, $\lambda(a) \leq 1 - \nu \cdot \rk_{\calB^*}(a)$. 
\end{itemize}
\end{Def}

Note that any set of generators $\calB^*$ for $\calA^*$ gives us some
values of $n,d,\mu$ and $\nu$. We would like $n, d$ to be large. Also,
applying the Spectral decay condition to $b \in \calB^*$, we see that  
$$  1 - \mu \leq \lambda(b) \leq  1 - \nu$$ 
hence $\mu \geq \nu$. Ideally, we would like
them to be within a constant factor of each other.

We refer to such graphs as ``derandomized hypercubes''. The reason is
that if there is a generating set of size $n$ which is significantly
larger than the dimension $h$, then the resulting
graph has spectral properties that resemble the $n$
dimensional hypercube, although it has only $2^h \ll 2^n$
vertices. Every Cayley graph $\Cay(\calA,\calD)$ together with a
generating set $\calB^*$ gives us a derandomized hypercube, the
parameters $n, d,\mu, \nu$ tell us how good the derandomization is
(just like any code $\calC$ and dual distribution $\calD$ gives us
local tester, whose quality is governed by the parameters it achieves).

\subsection{Derandomized hypercubes from Locally Testable Codes}

Barak et al. proposed the following construction of Derandomized Hypercubes from any Locally Testable Code \cite{BGH+12}.
Given $\calC$ which is an $[n,k,d]_2$ linear code with a local
tester $\calD$, they consider the Cayley graph $C(\calC^\perp,\calD)$
on $\calC^\perp \cong\F_2^{n-k}$   whose edge weights are distributed
according to  $\calD$.

\begin{Thm}
\label{thm:short-code}
Let $\calC$ be an $[n,k,d]_2$ linear code for $d \geq 3$, and let
$\calD$ be an $(\eps,\delta)$-tester for $\calC$. There
exists a $d$-wise independent set $\bcalE$ of size $n$ which is a $(2\eps,2\delta)$-spectrum generator for  $\Cay(\calC^\perp,\calD)$.
\end{Thm}
\begin{proof}
Observe that $(\calC^\perp)^* \cong \cosC$. 
This is because each  $v \in \calV$ defines a linear
function on $\calC^\perp$ given by $v(\alpha) = \alpha(v)$, and 
$v,v'$ define the same function iff the lie in the same coset
of $\calC$. 

We take $\bcalE = \{\bar{e}_1,\ldots,\bar{e}_n\}$ to be the cosets
corresponding to the received words $e_1,\ldots,e_n$. By Lemma \ref{lem:ltc-coset}, since $\calC$ has distance $d$,
the set $\bcalE$ is $d$-wise independence. We will show that it is a $(2\eps,2\delta)$-spectrum generator for
$\Cay(\calC^\perp,\calD)$.

We bound the eigenvalues
using the correspondence between the spectrum of $\Cay(\calC^\perp,\calD)$ and the soundness of the tester
$\calD$ established by Barak \etal. For $\bar{v} \in \cosC$, let $\chi_{\bar{v}}$ and $\lambda(\bar{v})$ denote the corresponding eigenvalue.  
\cite[Lemma 4.5]{BGH+12} says that
\begin{align}
\label{eq:eval}
\lambda(\bar{v}) = 1- 2\Rej(\bar{v},\calD).
\end{align}

\begin{itemize}
\item {\bf Smoothness implies Large Eigenvalues.} By Lemma
  \ref{lem:ltc-coset} the smoothness of $\calD$ implies $\Rej(\bar{e}_i,\calD) \leq \eps$. 
By Equation \ref{eq:eval}, 
\begin{align*}
\lambda(\bar{e}_i) = 1- 2\Rej(e_i,\calD) \geq 1 - 2\eps.
\end{align*}

\item {\bf Soundness implies Spectral decay.} 
Fix $\bar{v} \in \calV/\calC$ so that $\rk_{\calB^*}(\bar{v}) \geq
d'$. By Lemma \ref{lem:ltc-coset}, the soundness of $\calD$ implies
$\Rej(\bar{v},\calD) \geq \delta d'$.   
By Equation \ref{eq:eval}, 
\begin{align*}
\lambda(\bar{v})  = 1 - 2\Rej(\bar{v},\calD) \leq 1 - 2\delta d'.
\end{align*}
\end{itemize}

\end{proof}

\subsection{Locally Testable Codes from Derandomized Hypercubes}

We show how to start from a Cayley graph on $\calA = \F_2^h$ and a set
of generators for $\calA^*$ and get a locally testable code from
it. Our construction takes a Cayley graph $\Cay(\calA,\calD')$ and a
$(d,\mu,\nu)$-spectrum generator  $\calB^*$.

We define the locally testable code $\calC$ by specifying the dual
code $\calC^\perp$ and the tester $\calD$. We view elements $\alpha \in
\calA$ as  messages, and embed them into $\F_2^n$ using the map
\begin{align}
\label{eq:dual-code}
f(\alpha) = (b_1(\alpha),\ldots,b_n(\alpha)).
\end{align}
Since $\calB^*$ generates $\calA^*$, the mapping $f$ is injective. Its
image is a $h$-dimensional subspace of $\F_2^n$ which we denote by
$\calC^\perp$. The LTC will be $\calC$, which is the dual of $\calC^\perp$.
The distribution $\calD'$ on $A$ induces a distribution
$\calD = f(\alpha)_{\alpha \in \calD'}$ on $\calC^\perp$, which is the tester for $\calC$. 

\begin{Thm}
Let $\Cay(\calA = F_2^h,\calD')$ be a Cayley graph  and let $\calB^* =
\{b_1,\ldots,b_n\}$ be a $d$-wise independent $(\mu,\nu)$-spectrum
generator for it. Let $\calC$ be the dual of the code specified by
Equation \ref{eq:dual-code}. Then $\calC$ is an $[n, n-h,d]_2$ linear
code and $\calD$ is a $(\mu/2,\nu/2)$-tester for $\calC$.
\end{Thm}
\begin{proof}
It is clear that $\calC^\perp$ is an $[n,h]_2$ code, and hence $\calC$
is an $[n,n- h]_2$ code. Recall that $f:\calA \rgta \calC^\perp$ is an
isomorphism. Since $\F_2^n/\calC \cong (\calC^\perp)^*$, $f$ induces 
an isomorphism $g:\calA^* \rgta \F_2^n/\calC$, with
property that for $a \in \calA^*$ and $\alpha \in \calA$, 
\begin{align}
\label{eq:def-g}
g(a)(f(\alpha)) = a(\alpha).
\end{align}

We observe that $g(b_i) = \bar{e}_i$, since
\begin{align*}
\bar{e}_i(f(\alpha)) = e_i \cdot f( \alpha) = b_i(\alpha).
\end{align*}
Since $\calB^*$ is $d$-wise independent, so is $\bcalE$, which by
Lemma \ref{lem:ltc-coset} implies that $\calC$ has distance $d$.
This also implies that for any $a \in \calA^*$,
\begin{align*}
a= \sum_{i \in S}b_i \iff g(a)  = \sum_{i \in S}\bar{e}_i 
\end{align*}
Hence by Lemma \ref{lem:ltc-coset}, we have $d(g(a),\calC) =\rk_{\calB^*}(a)$.
We can now deduce the local testability of $\calC$ from the
spectral properties of $\calB^*$.

\begin{itemize}

\item {\bf Large Eigenvalues Imply Smoothness.}
In order to bound the smoothness of $\calD$ we need to bound
\begin{align*}
\Rej(\bar{e}_i,\calD) = \Pr_{\alpha \in \calD'}[\bar{e}_i\cdot
  f(\alpha) = 1] = \Pr_{\alpha \in \calD'}[b_i(\alpha)  = 1]
\end{align*}
We have
\begin{align*}
1 - \mu \leq \lambda(b_i) = \E_{\alpha \in
  \calD}[(-1)^{b_i(\alpha)}] = 1 -  2\Pr_{\alpha \in
  \calD}[b_i(\alpha) =1]
\end{align*}
which implies that 
\begin{align*}
\Rej(\bar{e}_i,\calD) \geq \frac{\mu}{2}.
\end{align*}

\item {\bf Spectral decay implies soundness.} 

Consider $\bar{v}\in \F_2^n/\calC$ such that $d(\bar{v},\calC)\geq
d'$. Let $\bar{v} = g(a)$ for $a\in \calA^*$, so that
$\rk_{\calB^*}(g(a)) \geq d'$. From the spectral decay
property of $\calB^*$,
\begin{align*}
1 - \nu d' \geq \lambda(a) = \E_{\alpha \in
  \calD'}[(-1)^{a(\alpha)}] = 1 -  2\Pr_{\alpha \in
  \calD'}[a(\alpha) =1] 
\end{align*}
hence
\begin{align*}
\Pr_{\alpha \in \calD'}[a(\alpha) =1] \geq \frac{\nu d'}{2}.
\end{align*}
The soundness of the tester follows by noting that
\begin{align*}
\Rej(\bar{v},\calD) = \Pr_{\alpha \in \calD'}[g(a)\cdot f(\alpha)
  =1] =  \Pr_{\alpha \in \calD'}[a(\alpha) =1]
\end{align*}
\end{itemize}
\end{proof}

\subsection{Some consequences of this equivalence}
\label{sec:sse-cons}

This equivalence lets us reformulate questions regarding LTCs as
questions regarding the existence of certain families of derandomized
hypercubes.  

\begin{Cor}
\label{cor:sse-asymptotic}
There exists an asymptotically good family of codes $\{\calC_n\}$
where $\calC_n$ has blocklength $n$ and is $(O(1/n),\Omega(1/n))$-locally testable iff for infintitely many $h$ there
exists a Cayley graph $\calG_h = \Cay(\F_2^h,\calD)$ and a set $\calB^*$ of generators for $(\F_2^h)^*$ such that
\begin{itemize}
\item the elements of $\calB^*$ are $(\rho_0h)$-wise independent for $\rho_0 >0$,
\item $|\calB^*| \geq \rho_1 h$ for $\rho_1 > 1$,
\item $\calB^*$ is an $(O(1/h),\Omega(1/h))$-spectrum generator for $\calG_h$.
\end{itemize}
\end{Cor}

\begin{Cor}
\label{cor:sse-constant}
Let $d \geq 3$. There exists an asymptotic family of
codes $\{\calC_n\}$ where $\calC_n$ has parameters $[n,n - c_d\log
  n,d]_2$ and is $(\eps,\Omega(\eps))$-locally testable iff for infinitely many $h$ 
there exists a Cayley graph $\calG_h = \Cay(\F_2^h,\calD)$ and a set $\calB^*$ of generators for $(\F_2^h)^*$ such that 
\begin{itemize}
\item the elements of $\calB^*$ are $d$-wise independent,
\item $|\calB^*| \geq 2^{h/c_d}$,
\item $\calB^*$ is an $(\eps,\Omega(\eps))$-spectrum generator for $\calG_h$.
\end{itemize}
\end{Cor}

Next we show that derandomized hypercubes are small-set expanders.
We say that a regular graph $G$ with $n$ vertices is a
$(\tau, \phi)$-expander if for every set $S$ of at most $\tau n$
vertices, at least a fraction $\phi$ of the edges incident to $S$
leave $S$ (i.e. are on the boundary between $S$ and
$\overline{S}$). 

The following lemma says that if a graph has a  $(\mu, \nu)$ spectrum generator, then it is a $(\tau, \phi)$-expander
for appropriately chosen $\tau$ and $\phi$. The lemma is proved in
\cite{BGH+12}. Since our terminology and notation is different, we
present a proof of the Lemma in Appendix \ref{app:hypercon}.

\begin{Lem}\cite{BGH+12}
\label{lem:hypercon-sse}
Let $G = \Cay(\calA,\calD)$ be a Cayley graph on the group $\calA =
\F_2^h$. Let $\calB^*$ be a $d$-wise independent set which is a
$(\mu,\nu)$-spectrum generator for $G$. Then $G$ is a $(\tau,\phi_\tau)$ 
expander for $\phi_\tau = \nu d/4 - 3^{d/2} \tau^{1/4}$.
\end{Lem}

To interpret the expansion bound, think of $\nu d/4 = \Omega(1)$ (we can
assume that the graphs obtained from LTCs have this property, since
this is analogous to saying that words at distance $d/4$ are rejected
with constant probability). So if  we take $\tau = \exp(-d)$, then $\phi_\tau = \Omega(1)$.
A particularly interesting instantiation of this bound is obtained by
combining Corollary~\ref{cor:sse-constant} and
Lemma~\ref{lem:hypercon-sse}: 
\begin{Cor}
\label{cor:sse-constant-graph-construction}
For $d \geq 3$, suppose there exists an asymptotic family of
codes $\{\calC_n\}$ where $\calC_n$ has parameters $[n,n - c_d\log
  n,d]_2$ and is $(O(1/d),\Omega(1/d))$-locally testable. There for
infinitely many $h$ there exists a Cayley graph $\calG_h =
\Cay(\F_2^h,\calD)$ such that $\calG_h$ is $(O(9^{-d}),
\Omega(1))$-expander and has $2^{h/c_d}$ eigenvalues greater than $1 -
O(1/d)$. 
\end{Cor}

In contrast, Arora \etal\ \cite{ABS10} showed that if $G$ is an
$(\tau, \Omega(1))$-expander, then there are at most
$n^{O(\epsilon)}/\tau$ eigenvalues greater than $1 - \epsilon$.  
Their bound implies that the graph $\calG_h$ obtained in
Corollary~\ref{cor:sse-constant-graph-construction} can have at most
$2^{O(h/d)}$ eigenvalues greater than $1 - O(1/d)$. If there exist
LTCs where $c_d = O(d)$, the resulting graphs $\calG_h$ would meet the
ABS bound. The only lower bound we know of for $c_d$ is $c_d \geq d/2$
by the Hamming bound.

\eat{
Combining Corollary~\ref{cor:sse-asymptotic} and Lemma~\ref{lem:hypercon-sse}, we get 
\begin{Cor}\label{cor:sse-asymptotic-graph-construction}
Suppose that there exists an asymptotically good family of codes $\{\calC_n\}$
where $\calC_n$ has blocklength $n$ and is $(O(1/n),\Omega(1/n))$-locally testable. Then there exists constants $\rho_0 > 0$ and $\rho_1 > 1$ such that for infinitely many $h$ there exists a Caylay graph $\calG_h = \Cay(\F_2^h, \calD)$ which is a $(9^{-\rho_0 h} \cdot \Omega(\rho_0)^4, \Omega(\rho_0))$-expander and has $2^{\frac{\rho_0 h}{2}}$ eigenvalues greater than $1 - O(\rho_0)$.
\end{Cor}
\begin{proof}
Let $\calG_h$ be the graph described in
Corollary~\ref{cor:sse-asymptotic}. The expansion property is directly
derived from Lemma~\ref{lem:hypercon-sse}.  

For the number of high eigenvalues, we observe that for every $d' \leq \frac{\rho_0 h}{2}$, all linear combinations of every $d'$ generators from $\calB^*$ are distinct, and correspond to eigenvectors with eigenvalue at least $1 - O(d'/ h)$. Now we choose $d' = \frac{\rho_0 h}{2}$, we know the number of eigenvalues that are at least $1 - O(\rho_0)$ is at least 
\[
{\rho_1 h \choose \frac{\rho_0 h}{2}} \geq \left(\frac{2\rho_1}{\rho_0}\right)^{\frac{\rho_0 h}{2}} \geq 2^{\frac{\rho_0 h}{2}} .
\]
\end{proof}
}

\section*{Acknowledgements}
We would like to thank Alex Andoni, Anupam Gupta and Kunal Talwar for
useful discussions and pointers to the literature on metric
embeddings. We also thank Raghu Meka, Prasad Raghavendra and Madhu
Sudan for helpful discussions. 

\bibliographystyle{alpha}
\bibliography{sse-ltc,pseudorandomness,salil}
\appendix
\section{Proof of Theorem \ref{thm:cov-radius}}
\label{app:boosting}

Let $\calD^{\oplus \ell}$ denote the distribution on $\calC^\perp$ where
we sample $\ell$ independent codewords according to $\calD$ and add
them. We claim that for suitable $\ell$, both $\eps$ and $\delta$
scale by roughly a factor of $\ell$.

\begin{Lem}
\label{lem:boost}
Let $\ell$ be such that 
\begin{align*}
\ell \leq \frac{1}{4\Rej(\bx,\calD)} \  \ \text{for all} \  \ \bx \in \F_2^n/\calC.
\end{align*}
Then $\calD^{\oplus \ell}$ is an $(\eps\ell, \delta \ell/2)$-tester for $\calC$.
\end{Lem}
\begin{Proof}
The tester $\calD^{\oplus \ell}$ is an $\ell\eps$-smooth tester by the
union bound. Its soundness can be analyzed by noting that
\begin{align*} 
1 - 2\Rej(\bx,\calD^{\oplus \ell}) & = (1 - 2\Rej(\bx,\calD))^{\ell}
\end{align*}
Using the bound on $\ell$ to truncate the RHS, we get
\begin{align}
\label{eq:boost}
\Rej(\bx,\calD^{\oplus \ell}) \geq \frac{\ell}{2}\Rej(\bx,\calD) \geq
\frac{\ell \delta}{2}d(\bx,0).
\end{align}
\end{Proof}

We use this to prove Theorem \ref{thm:cov-radius}.

\begin{proof}[Proof of Theorem~\ref{thm:cov-radius}]
We start with an $(\eps,\delta)$-tester where $\delta \geq \eps/c_1(\calG)$.
If $\delta$ exceeds the claimed bound, we are already done. Assume this
is not true, so
\[ \delta \leq \frac{1}{16c_1(\calG)t}, \ \eps \leq c_1(\calG)\delta
\leq \frac{1}{16t} \]

Since the covering radius is $t$, we have that for every $\bx \in \F_2^n$,
\[ \Rej(\bx,\calD) \leq t\eps \leq 1/16\]
Let $\ell = \lfloor 1/(4t\eps)\rfloor$ so that
\[ \frac{1}{8t\eps} \leq \ell \leq \frac{1}{4t\eps}. \]

By Lemma \ref{lem:boost}, $\calD^{\oplus \ell}$ has smoothness $\eps'$ where
\[ \eps' \leq \ell \eps \leq \frac{1}{4t}\]
and soundness $\delta'$ where
\[ \delta' \geq \frac{1}{2}\ell\delta \geq
\frac{1}{2}\frac{1}{8t\eps}\frac{\eps}{c_1(G)} \geq \frac{1}{16tc_1(G)}.  \]
\end{proof}

\section{Proof of Lemma \ref{lem:hypercon-sse}}
\label{app:hypercon}

We first show the follow hypercontractive inequality:

\begin{Claim}\label{clm:hypercon}
Let $B = \{b_1,\ldots, b_n\} \subseteq \F_2^h$ be $(4d+1)$-wise independent. For every function $f : \F_2^h \to \R$ defined as
\[
f(x) = \sum_{S \subseteq [n], |S| \leq d} \hat{f}(S) \prod_{i \in S} \chi_{b_i}(x) ,
\]
we have 
\[
\E_x [f(x)^4] \leq 9^d \left(\E_x [f(x)^2]\right)^2 .
\]
\end{Claim}
\begin{Proof}
Let $g : \F_2^n \to \R$ be 
\[
g(y) = \sum_{S \subseteq [n], |S| \leq d} \hat{f}(S) \prod_{i \in S} y_i .
\]
The statement is proved by the standard $(2, 4)$-hypercontractive inequality (applied to $g$ function) and the observation that
\[
\E_x[f(x)^2] = \E_y [g(y)]^2 = \sum_{S \subseteq [n], |S| \leq d} \hat{f}(S)^2 ,
\]
and
\begin{align*}
\E_x[f(x)^4] = \E_y[g(y)^4] = \sum_{\stackrel{|S_1|, |S_2|, |S_3|, |S_4| \leq d}{S_1 \Delta S_2 \Delta S_3 \Delta S_4 = \emptyset}} \hat{f}(S_1)\hat{f}(S_2)\hat{f}(S_3)\hat{f}(S_4) .
\end{align*}
\end{Proof}

We now proceed to prove Lemma \ref{lem:hypercon-sse}.
\begin{proof}[Proof of Lemma \ref{lem:hypercon-sse}]
For any two functions $f, g : \F_2^h \to \R$, define their inner-product as
\[ \langle f, g \rangle = \E_{x \in \F_2^h} [f(x) g(x)]\] 
and the $p$-norm of $f$ to be 
\[ \|f\|_p = \left(\E_{x \in \F_2^h} f(x)^p \right)^{1/p}.\] 
For every function $f : \F_2^h \to \R$ with Fourier expansion 
\[ f(x) = \sum_{a \in \F_2^h} \hat{f_a} \chi_{a}(x)\]
let 
\begin{align*}
f^{<d/4}(x) & = \sum_{a : \rk_{\calB^*}(a) < d/4} \hat{f_a} \chi_{a}(x),\\ 
f^{\geq d/4}(x) & = \sum_{a : \rk_{\calB^*}(a) \geq d/4} \hat{f_a} \chi_{a}(x) .
\end{align*}
Fix a set $\calS \subseteq \F_2^h$. Let $\tau = \mu(\calS)$ be the volume of $\calS$. Let $\bm{1}_{\calS}(x) = \bm{1}_{x \in \calS}$ be the indicator function $\calS$. Note that $\|\bm{1}_{\calS}\|_p^p = \tau$ for every $p \geq 1$. We will lower bound the expansion $\Phi(\calS) = 1 - \langle \bm{1}_{\calS}, G\bm{1}_{\calS}\rangle/\tau$, which is the fraction of the edges incident to $\calS$ leaving $\calS$. Observe that
\begin{align}
\langle \bm{1}_{\calS}, G\bm{1}_{\calS}\rangle = \langle \bm{1}_{\calS}, G\bm{1}_{\calS}^{< d/4}\rangle +  
\langle \bm{1}_{\calS}, G\bm{1}_{\calS}^{\geq d/4}\rangle. \label{eq:hypercon-sse}
\end{align}
The first term in the $\mathrm{RHS}$ of \eqref{eq:hypercon-sse} is upper bounded as
\[
\langle \bm{1}_{\calS}, G\bm{1}_{\calS}^{< d/4}\rangle  = \|\bm{1}_{\calS}\|_{4/3} \|G\bm{1}_{\calS}^{< d/4}\|_4 \leq \|\bm{1}_{\calS}\|_{4/3} \cdot \sqrt{3}^{d} \|\bm{1}_{\calS}\|_{2} = \sqrt{3}^{d} \tau^{5/4}
\]
by H\"older's inequality and Claim~\ref{clm:hypercon}. The second term in the $\mathrm{RHS}$ of \eqref{eq:hypercon-sse} is upper bounded as
\[ 
\langle \bm{1}_{\calS}, G\bm{1}_{\calS}^{\geq d/4}\rangle \leq (1 - \nu d/4) \|\bm{1}_{\calS}\|_2^2 =  (1 - \nu d/4)\tau. \] 
In all, we have 
\[
\Phi(\calS) = 1 - \frac{\langle \bm{1}_{\calS}, G\bm{1}_{\calS}\rangle}{\tau} \geq 1 - \sqrt{3}^d \tau^{1/4} - (1 - \nu d/4) = \nu d/4 - \sqrt{3}^d \tau^{1/4} .
\]
\end{proof}

\end{document}